\documentclass[11pt]{article}

\usepackage{amssymb,amsmath,amsfonts,amsthm}

 \usepackage[numbers,sort&compress]{natbib}

\usepackage{hyperref}
\hypersetup{
   colorlinks   =  true,
    linkcolor    = blue,
    citecolor    = red,
     urlcolor	=blue,
 }

\newcommand\be{\begin{equation}}
\newcommand\ee{\end{equation}}
\newcommand\bea{\begin{eqnarray}}
\newcommand\eea{\end{eqnarray}}

\theoremstyle{plain}
\newtheorem{theorem}{Theorem}

\newtheorem{proposition}[theorem]{Proposition}

\theoremstyle{definition}

\numberwithin{theorem}{section}
\numberwithin{equation}{section}

\def\dd{{\rm d}}
\def\eee{{\rm e}}
\def\>#1{{\bf #1}}

 \def\gam{\gamma}

\begin{document}

\
  \vskip0.5cm

 \noindent
 {\Large \bf  Solutions by quadratures of complex Bernoulli differential equations\\[4pt] and their quantum deformation}

\medskip
\medskip
\medskip

\begin{center}

{\sc  Rutwig Campoamor-Stursberg$^{1,2}$, Eduardo Fern\'andez-Saiz$^{3}$\\[2pt] and Francisco J.~Herranz$^4$}

\end{center}

\medskip

 \noindent
$^1$ Instituto de Matem\'atica Interdisciplinar, Universidad Complutense de Madrid, Plaza de Ciencias 3, E-28040 Madrid,  Spain

\noindent
$^2$ Departamento de \'Algebra, Geometr\'{\i}a y Topolog\'{\i}a,  Facultad de Ciencias
Matem\'aticas, Universidad Complutense de Madrid, Plaza de Ciencias 3, E-28040 Madrid, Spain

\noindent
{$^3$ Department of Quantitative Methods, CUNEF Universidad,   E-28040  Madrid, Spain}

\noindent
{$^4$ Departamento de F\'isica, Universidad de Burgos,
E-09001 Burgos, Spain}

 \medskip

\noindent  E-mail: {\small
 \href{mailto:rutwig@ucm.es}{rutwig@ucm.es}, \href{mailto:e.fernandezsaiz@cunef.edu}{e.fernandezsaiz@cunef.edu}, \href{mailto:fjherranz@ubu.es}{fjherranz@ubu.es}
}

\medskip

\begin{abstract}
\noindent It is shown that the complex Bernoulli differential equations admitting the supplementary structure of a Lie--Hamilton system related to the book algebra $\mathfrak{b}_2$ can always be solved by quadratures, providing an explicit solution of the equations. In addition, considering the quantum deformation of Bernoulli equations, their canonical form is obtained and an exact solution by quadratures is deduced as well. It is further shown that the approximations of $k^{th}$-order in the deformation parameter from the quantum deformation are also integrable by quadratures, although an explicit solution cannot be obtained in general. Finally, the multidimensional quantum deformation of the  book Lie--Hamilton systems is studied, showing that, in contrast to the multidimensional analogue of the undeformed system, the resulting system is coupled in a non-trivial form.
  \end{abstract}
\medskip
\medskip


\medskip


\medskip

 \newpage

\tableofcontents


\section{Introduction}
\label{s1}

The obtainment of explicit solutions to (ordinary or partial) differential equations, as well as finding effective criteria that determine whether such solutions can be found at all, is certainly one of the central research problems, not only for the structural theory of differential equations, but also for their applications, notably for dynamical systems. The beginning of the systematic study and classification of differential equations goes back to the XIX century, in an analytical context, where the pioneering work of Goursat, Picard or Painlev\'e, among others, emerged~\cite{Gra,Page,Pain}, complementing the purely geometrical approach initiated by Poincar\'e, that has evolved to constitute one of the most important current techniques~\cite{Arn,Dua,Chandra2}. A somewhat different ansatz, based on group theory, and taking into account the symmetry properties, was developed by Lie, leading to the modern symmetry method, that provides a formal explanation for several of the already known solution methods~\cite{Lie,Dik,Ince}. Indeed, the Lie symmetry method is an effective tool for either reducing the order of an equation, to find a canonical representative of the equation, and furthermore, to determine criteria that ensure that the equation is linearizable. In this context, several procedures have been proposed to linearize and solve systems of differential equations, as well as to guarantee their integrability~\cite{Nuc,Lak}. One auxiliary tool, whenever an explicit solution is either not available or excessively cumbersome, is to determine the existence of a (nonlinear) superposition principle, thus reducing the problem to finding a certain number of particular solutions that, in combination with some significant constants, allows to write the general solution. A characterization of this property was obtained by Lie himself in~\cite{LSc} using group theoretical methods, hence establishing a technique that has been extended by several authors in different directions (see e.g.~\cite{Vess,Wei,PW,Bun,Reid,Car,CGL,Grun} and references therein). In some circumstances, even if it is known that an equation admits a superposition principle, it may go unnoticed that, after an appropriate local diffeomorphism, the differential equation admits an explicit solution in the new coordinates.

Consider for instance a complex function $w(t)=u(t)+ {\rm i} v(t)$ and the first-order equation
\begin{equation}\label{be2}
\frac{{\rm d}w}{{\rm d}t} =f(t)w+ g(t)w^n
\end{equation}
for $n=2$ and arbitrary real functions $f(t),\, g(t)$, which is a special case of the so-called complex Bernoulli equation~\cite{Zoladek}. Separating the real and imaginary parts, we are led to the first-order system
\begin{equation}\label{be1}
\begin{split}
\frac{{\rm d}u}{{\rm d}t}& =f(t)u +g(t)(u^2-v^2),\\[2pt]
\frac{{\rm d}v}{{\rm d}t}& =f(t)v+ 2g(t)uv.
\end{split}
\end{equation}
Solving the second equation with respect to $u$ we are led to
\begin{equation}
u=\frac{\frac{{\rm d}v}{{\rm d}t}-f(t)v}{2g(t)v},
\end{equation}
and the second-order equation reads as
\begin{equation}\label{be3}
\frac{{\rm d}^2v}{{\rm d}t^2} =\frac{3}{2v}\left(\frac{{\rm d}v}{{\rm d}t}\right)^2+ \frac{1}{g(t)}\frac{{\rm d}g}{{\rm d}t}\frac{{\rm d}v}{{\rm d}t}+\left( \frac{{\rm d}f}{{\rm d}t}-\frac{1}{2}f(t)^2-\frac{f(t)}{g(t)}\frac{{\rm d}g}{{\rm d}t}\right)v-2g(t)^2v^3.
\end{equation}
By means of the change of the dependent variable $v=\xi^{-2}$, the latter equation reduces to
\begin{equation}\label{be3a}
\frac{{\rm d}^2\xi}{{\rm d}t^2} =\frac{1}{g(t)}\frac{{\rm d}g}{{\rm d}t}\frac{{\rm d}\xi}{{\rm d}t}-\frac{1}{2}\left( \frac{{\rm d}f}{{\rm d}t}-\frac{1}{2}f(t)^2-\frac{f(t)}{g(t)}\frac{{\rm d}g}{{\rm d}t}\right)\xi+\frac{g(t)^2}{\xi^3}.
\end{equation}
This equation, which is not linear, can be shown to possess a Lie point symmetry algebra isomorphic to $\mathfrak{sl}(2,\mathbb{R})$~\cite{C131} (equation (\ref{be3a}) actually appears as a perturbation of linear homogeneous ODE preserving  an $\mathfrak{sl}(2,\mathbb{R})$-subalgebra of Noether symmetries), and by application of the Lie symmetry method~\cite{Olv}, it can be reduced to a linear equation by means of a non-point transformation~\cite{Le3}, from which the general solution of (\ref{be2}) is deduced. However, this procedure cannot be applied directly, without enormous computational complication, for integer values $n\geq 3$, as the real and imaginary parts are not separable:
\begin{equation}\label{be4}
\begin{split}
\frac{{\rm d}u}{{\rm d}t}& =f(t)u +g(t)\sum_{j=1}^{\left[\frac{n}{2}\right]+1}\frac{(-1)^{j+1} n!}{(2j-2)!(n+2-2j)!}u^{n+2-2j}v^{2j-2},\\[0.1cm]
\frac{{\rm d}v}{{\rm d}t}& =f(t)v+ g(t)\sum_{j=1}^{\left[\frac{n}{2}\right]+1}\frac{(-1)^{j+1} n!}{(2j-1)!(n+1-2j)!}u^{n+1-2j}v^{2j-1}.
\end{split}
\end{equation}
The system, however, is endowed with an additional structure that simplifies its analysis. It can be easily seen that the above system can be written in terms of the time-dependent vector field
\begin{equation}
{\bf X}= f(t){\bf X}_1+ g(t){\bf X}_2
\end{equation}
with the $t$-independent vector fields given by
\begin{equation}\label{xxx}
\begin{split}
{\bf X}_1&=u\frac{\partial}{\partial u}+v\frac{\partial}{\partial v},\\[4pt]
{\bf X}_2&=\sum_{j=1}^{\left[\frac{n}{2}\right]+1}\frac{(-1)^{j+1} n!}{(2j-2)!(n+2-2j)!}u^{n+2-2j}v^{2j-2}\frac{\partial}{\partial u}\\[2pt]
&\qquad\quad+
\sum_{j=1}^{\left[\frac{n}{2}\right]+1}\frac{(-1)^{j+1} n!}{(2j-1)!(n+1-2j)!}u^{n+1-2j}v^{2j-1}\frac{\partial}{\partial v} ,
\end{split}
\end{equation}
satisfying the commutator
\be
\left[{\bf X}_1,{\bf X}_2\right]=(n-1){\bf X}_2.
\ee
This implies that (\ref{be4}) possesses the structure of a Lie system (see e.g.~\cite{LSc,LuSa} and references therein) with a Vessiot--Guldberg algebra isomorphic to the so-called book algebra $\mathfrak{b}_2$~\cite{LH2015, BHLS, Ballesteros6, BCFHL, CFH}. Lie systems are characterized by the remarkable property of admitting a (nonlinear) superposition principle. As will be deduced later, the system (\ref{be4}) admits a supplementary structure, namely that of a Lie--Hamilton system (LH in short), that allows us to reduce it to a linear system, from which an exact solution by quadratures will be obtained. Actually, as shown in \cite{BHLS}, the complex Bernoulli equation (\ref{be2}) with complex coefficients admits the structure of a Lie system, but not of an LH system for arbitrary (complex) choices of the coefficient functions $f,g$. Hence, if an LH is known to be solvable by quadratures for a (local) system of coordinates, then, for any (local) diffeomorphism, the transformed system will also be solvable by quadratures. This fact will enable us to obtain an exact solution of  (\ref{be4}) for arbitrary values of $n$ and real coefficients $f(t),\, g(t)$.

LH systems constitute a natural generalization of Lie systems, with the salient property of being related to dynamical systems through a Hamiltonian \cite{LuSa,LH2015}. Precisely this fact provides a systematic procedure to compute the constants of the motion needed for finding a superposition principle, using the so-called coalgebra formalism, a technique developed specifically in the context of (quantum) Hamiltonian systems \cite{BBHMR09}, hence valid for LH systems, but not applicable to generic Lie systems. For this reason, LH systems have been studied focusing mainly on their coalgebra symmetry \cite{BCFHL,Ballesteros6}. However, depending on the specific realization of the LH system in terms of Hamiltonian vector fields, suitable coordinate frames may be found that allow us either to linearize or to integrate directly the system, making the explicit construction of superposition rules superfluous. First results in this new direction,  that amounts to analyze the equivalent realizations by Hamiltonian vector fields associated to a given Lie algebra, were obtained in \cite{CFH} for some particular types of LH systems based on $\mathfrak{b}_2$. 

 This work is structured as follows. In Section~\ref{s2}, we review the general properties of LH systems associated to the Lie algebra $\mathfrak{b}_2$ and reconsider their exact integrability. As a new result, we determine the most general form of first-order systems equivalent to these $\mathfrak{b}_2$-LH systems by local diffeomorphisms and, using a basis different from that considered in the classification of LH systems~\cite{LH2015}, we obtain explicit solutions. In Section~\ref{s3}, we focus on complex Bernoulli equations admitting the structure of an LH system, and compute their explicit solution, completing and expanding the results of \cite{Ballesteros6}, where no exact solutions were provided. 

Section~\ref{s4} deals with the quantum deformation of $\mathfrak{b}_2$-LH systems \cite{BCFHL, Ballesteros6, CFH}, for which a solution by quadratures is computed, comparing the results with the explicit solvability of the approximations of $k^{th}$-order in the quantum deformation parameter, which have not been considered previously in the literature. The latter general results for deformed $\mathfrak{b}_2$-LH systems are applied in Section~\ref{s5} to the construction of the corresponding deformed complex Bernoulli differential equations and, furthermore, to the obtention of their exact solution. In Section~\ref{s6} we show that, for multidimensional quantum deformations, the structure of the resulting system is radically different from the undeformed case, although even in this case, the resulting equations can still be solved by quadratures. 

Finally, in Section~\ref{s7} we draw some conclusions and comment on possible extensions or generalizations of our results to the other isomorphism classes of LH systems, as well as their potential use in real-world applications.


\section{Lie--Hamilton systems on the book algebra revisited}
\label{s2}

We  briefly review the main features of LH systems based on the two-dimensional book Lie algebra $\mathfrak b_2$ (details concerning the general formalism of LH system can be found in \cite{LH2015,BHLS,Ballesteros6, CFH} and references therein). Consider the basis $\mathfrak b_2={\rm span}\{ v_A, v_B\}$   with Lie bracket
\be
[v_A,v_B] = - v_B .
\label{b1}
\ee
The generator $v_A$ can be seen as a  dilation operator, while $v_B$ can be interpreted as a translation operator.  As a common property to any Lie algebra~\cite{CP, Abe}, $\mathfrak b_2$ can be endowed with a trivial Hopf algebra structure through the (primitive) coproduct $\Delta : \mathfrak b_2 \to \mathfrak b_2 \otimes \mathfrak b_2 $ defined by
\be
\Delta (v_i)=v_i \otimes 1 + 1\otimes v_i  ,\qquad i=A,B.
\label{b2}
\ee
This map further defines an algebra homomorphism for arbitrary Lie algebra $\mathcal A$, satisfying the coassociativity constraint
\be
({\rm Id}\otimes\Delta)\Delta(a)=(\Delta\otimes {\rm Id})\Delta(a),\qquad \forall a\in \mathcal A.
\label{b2a}
\ee
Hence, the pair $(\mathcal A,\Delta)$ defines a coalgebra.

 We now consider the symplectic representation $D$ of $\mathfrak b_2$
defined by (see e.g.~\cite{Ballesteros6})
\be
h_A:= D(v_A)=  xy,\qquad h_B:= D(v_B)=  - x ,
\label{b4}
\ee
in terms of the canonical variables $(x,y)$ and with respect to the canonical symplectic form
\be
\omega = \dd x\wedge  \dd y.
\label{b3}
\ee
The   Hamiltonian functions (\ref{b4}) satisfy the following Poisson bracket with respect  to $\omega$:
\be
\{ h_A,h_B\}_\omega= - h_B .
\label{b5}
\ee
A realization of the $\mathfrak b_2$-generators in terms of vector fields with Cartesian coordinates $(x,y)\in \mathbb R^2$ is directly deduced from (\ref{b4}), using the inner product associated to $\omega$:
\be
\iota_{{\bf X}_i}\omega={\rm d}h_i ,\qquad  i=A, B,
\label{b6}
\ee
from which the vector fields
 \be
     \>X_A=x\, \frac{\partial}{\partial x}- y\, \frac{\partial}{\partial y}, \qquad    \>X_B=\frac{\partial}{\partial y},\qquad  [\>X_A,\>X_B]=\>X_B ,
\label{b7}
\ee
result. It is straightforward to verify that the invariance condition
\be
{\cal L}_{\mathbf{X}_i}\omega =0, \qquad i=A,B,
\label{b7b}
\ee
is satisfied.
The Hamiltonian functions (\ref{b4}) and the vector fields (\ref{b7}) determine a $t$-dependent Hamiltonian (respectively a $t$-dependent vector field) depending on two real arbitrary parameters $b_A(t)$ and $b_B(t)$ and given by
\begin{equation}
\begin{split}
h&=b_A(t)  h_A +b_B(t)  h_B   =  b_A(t)\,  xy - b_B(t) \,x   , \\[2pt]
  \>X&=b_A(t)   \>X_A +b_B(t)   \>X_B   =  b_A(t) \left( x\, \frac{\partial}{\partial x}- y\, \frac{\partial}{\partial y} \right)    + b_B(t) \, \frac{\partial}{\partial y}  .
\end{split}
 \label{b8}
\end{equation}
Either of these expressions leads to the same (linear) system of non-autonomous  ODEs on $\mathbb R^2$, namely
   \begin{equation}
\frac{\dd x}{\dd t}=b_A(t)\,  x ,
\qquad
\frac{\dd y}{\dd t}=-b_A(t)\, y+ b_B(t).
 \label{b9}
\end{equation}
This system is easily solved by quadratures, with its exact solution being given by
  \be
\begin{split}
x(t)=c_1 \,\eee^{\gam(t)},\qquad y(t)=\left( c_2 +  \int_a^t  \eee^{\gam(u )} b_B(u) \dd u\right)  \eee^{-\gam(t)},\qquad \gam(t):= \int_a^t b_A(s)\dd s ,
\end{split}
\label{b10}
\ee
where $c_1$ and $c_2$ are integration constants determined by the  initial conditions and $a$ is   an appropriate real number that ensures the existence of the integrals over the compact interval $[a,t]$.

The relations (\ref{b7}) clearly show that (\ref{b9}) determines a Lie system~\cite{LSc,PW,LuSa} with associated Vessiot--Guldberg Lie algebra isomorphic to $\mathfrak b_2$. The existence of Hamiltonian vector fields with respect to a Poisson structure implies that the system has a richer geometrical structure, namely that of an LH  system (see~\cite{LH2015,BHLS} and references therein).

In spite of the simplicity of the equations (\ref{b9}) in Cartesian coordinates, LH systems associated to the book algebra $\mathfrak{b}_2$ can appear in disguised form, due to an inappropriate choice of coordinates, which eventually makes it difficult to recognize them as   LH systems, as well as to find a  suitable integration strategy.

Let $\Phi:U\subset\mathbb{R}^2\rightarrow \mathbb{R}^2$ be a (local) diffeomorphism with
\begin{equation}
x= \varphi_1(u,v),\qquad y=\varphi_2(u,v).
\end{equation}
In the new coordinates, we have that
\begin{equation}
{\rm d}x\wedge {\rm d}y= \left(\varphi_{1,u}\varphi_{2,v}-\varphi_{1,v}\varphi_{2,u}\right) {\rm d}u\wedge {\rm d}v=\Xi \, {\rm d}u\wedge {\rm d}v,
\end{equation}
where, as a shorthand notation we have defined
\be
\Xi = \varphi_{1,u}\varphi_{2,v}-\varphi_{1,v}\varphi_{2,u} ,\qquad \varphi_{i,u}=\frac{\partial \varphi_i}{\partial u},\qquad \varphi_{i,v}=\frac{\partial \varphi_i}{\partial v} ,\qquad i=1,2.
\label{xii}
\ee
And the Hamiltonian functions (\ref{b4})  are transformed as follows:
\begin{equation}
\widehat{h}_1=\Phi(h_A) =\Phi(xy)  =\varphi_1(u,v)\varphi_2(u,v),\qquad \widehat{h}_2=\Phi(h_B)=\Phi(-x)=-\varphi_1(u,v).
\end{equation}
A routine computation shows that the transformed Hamiltonian vector fields are given by
\begin{equation}
\begin{split}
{\bf X}_1={\rm d}\Phi({\bf X}_A)&=\frac{\varphi_{1,v}\varphi_{2}+\varphi_{1}\varphi_{2,v}}{\Xi}\,\frac{\partial}{\partial u}-
\frac{\varphi_{1,u}\varphi_{2}+\varphi_{1}\varphi_{2,u}}{\Xi} \, \frac{\partial}{\partial v},\\[2pt]
{\bf X}_2={\rm d}\Phi({\bf X}_B)&=-\frac{\varphi_{1,v}}{\Xi} \, \frac{\partial}{\partial u}+\frac{\varphi_{1,u}}{\Xi}\,\frac{\partial}{\partial v},
\end{split}
\end{equation}
obviously satisfying the invariance condition (\ref{b7b}) and the commutator $\left[{\bf X}_1,{\bf X}_2\right]={\bf X}_2$. The $t$-dependent vector field ${\bf X}= b_A(t){\bf X}_1+b_B(t){\bf X}_2$ then leads to the system in the coordinates $(u,v)$
\begin{equation}\label{syta}
\begin{split}
\frac{{\rm d}u}{{\rm d}t}&=\frac{b_A(t)}{\Xi}\left(\varphi_{1,v}\varphi_{2}+\varphi_{1}\varphi_{2,v}\right)-\frac{b_B(t)}{\Xi}\varphi_{1,v},\\[2pt]
\frac{{\rm d}v}{{\rm d}t}&=-\frac{b_A(t)}{\Xi}\left(\varphi_{1,u}\varphi_{2}+\varphi_{1}\varphi_{2,u}\right)+\frac{b_B(t)}{\Xi}\varphi_{1,u} .
\end{split}
\end{equation}
Independently on the particular shape of the resulting functions, the general solution of the system (\ref{syta}) can be obtained explicitly from (\ref{b10}), using that $\Phi$ is (locally) invertible. Supposed that $\Psi=\Phi^{-1}$ with $u=\psi_1(x,y)$, $v=\psi_2(x,y)$, the solution of (\ref{syta})  is given in terms of the integrals in (\ref{b10}) as
\begin{equation}
u(t)=\psi_1(x(t),y(t)),\qquad v(t)=\psi_2(x(t),y(t)).
\end{equation}

 It should be mentioned that, following the classification of planar LH systems \cite{LH2015}, the Lie algebra $\mathfrak{b}_2$ admits only one equivalence class of realizations as vector fields on the real plane \cite{Kam}. This means that any first-order system that possesses the structure of a Lie system with Vessiot--Guldberg algebra isomorphic to $\mathfrak{b}_2$ is locally diffeomorphic to the system (\ref{b9}), and hence integrable by quadratures. This result can be stated as follows:

 \begin{proposition}\label{PRO2}
Any first-order system
\begin{equation}\label{lins}
\frac{\dd u}{\dd t}=\Theta(t,u,v),\qquad \frac{\dd v}{\dd t}=\Omega(t,u,v)
\end{equation}
such that the associated $t$-dependent vector field \be{\bf X}=\Theta(t,u,v)\frac{\partial}{\partial u}+\Omega(t,u,v)\frac{\partial}{\partial v}\ee
admits the decomposition
\begin{equation}
{\bf X}= f(t){\bf X}_1+g(t){\bf X}_2
\end{equation}
for some functions $f(t),g(t)$ and vector fields ${\bf X}_1,{\bf X}_2$ satisfying
\begin{equation}
\frac{\partial {\bf X}_1}{\partial t}=\frac{\partial {\bf X}_2}{\partial t}=0,\qquad \left[{\bf X}_1,{\bf X}_2\right]=\lambda {\bf X}_1+\mu{\bf X}_2\neq 0
\end{equation}
with $\lambda,\mu$ constants, can be solved by quadratures.
\end{proposition}

\begin{proof}
It suffices to verify that the vector fields ${\bf X}_1,{\bf X}_2$ span a Lie algebra isomorphic to $\mathfrak{b}_2$, the rest following from the previous discussion. If either $\lambda=0$ or $\mu=0$, a scaling transformation shows that ${\bf X}_1,{\bf X}_2$ satisfy the commutator (\ref{b7}), hence the Lie algebra is isomorphic to $\mathfrak{b}_2$. On the contrary, if $\lambda\mu\neq 0$, considering the new vector field $Y_2=\lambda {\bf X}_1+\mu{\bf X}_2$, we get $\left[{\bf X}_1,{\bf Y}_2\right]=\mu\left(\lambda {\bf X}_1+\mu{\bf X}_2\right)=\mu {\bf Y}_2$, showing that the Lie algebra is isomorphic to $\mathfrak{b}_2$. As Lie system (indeed, as LH system), (\ref{lins}) can always be rewritten as
\begin{equation}
\begin{split}
\Theta(t,u,v)= f(t) X_1^1(u,v)+g(t)X_2^1(u,v),\qquad \Omega(t,u,v)= f(t) X_1^2(u,v)+g(t)X_2^2(u,v),
\end{split}
\end{equation}
where \be {\bf X}_k=X_k^1(u,v)\frac{\partial }{\partial u}+X_k^2(u,v)\frac{\partial }{\partial v},\qquad k=1,2.\ee
 \end{proof}

We observe that, within the classification of  LH systems on $\mathbb R^2$ of \cite{LH2015}, the book algebra $\mathfrak{b}_2$ corresponds to the class I$_{14A}^{r=1}\simeq \mathbb{R} \ltimes \mathbb{R}$, although expressed in a different basis that is computationally more cumbersome than that considered in this work. Although, in general, it is a routine task to verify whether a first-order system (\ref{lins}) admits $\mathfrak{b}_2$ as Vessiot--Guldberg algebra, the effective determination of the local diffeomorphism is far from being trivial. It is precisely the additional structure of an LH system with the symplectic form and the invariance condition which allows us to systematically compute the appropriate change of coordinates, and hence to explicitly solve the system.

As an illustrative example, let us consider the system (\ref{be1}) from this perspective, instead of solving or transforming the second-order ODE (\ref{be3}). Clearly, the system is given by ${\bf X}= f(t){\bf X}_1+g(t){\bf X}_2$, with \be \displaystyle {\bf X}_1=u\frac{\partial}{\partial u}+v\frac{\partial}{\partial v},\qquad \displaystyle {\bf X}_2=(u^2-v^2)\frac{\partial}{\partial u}+2uv\frac{\partial}{\partial v},\qquad [{\bf X}_1,{\bf X}_2]={\bf X}_2 ,
\ee   showing that the system (\ref{be1}) is locally diffeomorphic to (\ref{b9}). The noncanonical symplectic form in the $(u,v)$ coordinates is given by \be \omega =\displaystyle \frac{1}{v^2}{\rm d}u\wedge {\rm d}v,\ee  from which we deduce the Hamiltonian functions
\begin{equation}
\begin{split}
\iota_{{\bf X}_1}\omega={\rm d}h_1={\rm d}\left(-\frac{u}{v}\right),\qquad \iota_{{\bf X}_2}\omega={\rm d}h_2={\rm d}\left(-\frac{u^2+v^2}{v}\right),
\end{split}
\end{equation}
with associated Poisson bracket $\left\{h_1,h_2\right\}_{\omega}=\omega\left({\bf X}_1,{\bf X}_2\right)=-h_2$. It is straightforward to verify that, for the parameter functions, we have the identities $f(t)=b_A(t)$ and $g(t)=b_B(t)$. Now, the comparison of the Hamiltonian and symplectic forms in both systems of coordinates allows us to compute easily the (local) change of coordinates, which is given by
\begin{equation}
x=\frac{u^2+v^2}{v},\quad y=-\frac{u}{u^2+v^2}\quad \iff \quad u=-\frac{x^2y}{1+x^2y^2},\quad v=\frac{x}{1+x^2y^2}.
\end{equation}
Inserting (\ref{b10})  into the preceding expressions provides the exact solution of  (\ref{be1}) in the coordinates $u,v$, without requiring further computation. This shows that, at least for the case of $\mathfrak{b}_2$, the main characteristics of LH systems, namely the coalgebra formalism for the construction of constants of the motion and the obtainment of a superposition principle, are not required, as the system can be solved explicitly by quadratures. It should be observed that, as the Lie algebra $\mathfrak{b}_2$ has no nontrivial Casimir invariant, the superposition principles considered in \cite{Ballesteros6} were deduced considering an embedding of the book algebra $\mathfrak{b}_2$ into the oscillator algebra $\mathfrak{h}_4$.

In this framework, it is worth mentioning that (systems of) differential equations that correspond to LH systems based on $\mathfrak{b}_2$ have been considered in various contexts, such as the (generalised) Buchdahl equations arising in the study of relativistic fluids \cite{Buchdahl,Chandrasekar} or some particular  Lotka--Volterra systems with $t$-dependent coefficients; they have been shown to be LH systems in \cite{LH2015,BHLS}. In addition,   $\mathfrak{b}_2$-LH systems that generalize time-dependent epidemic models with stochastic fluctuations have recently been  developed in \cite{CFH}. Nevertheless, an explicit diffeomorphism has only been deduced in the latter work, as well as in \cite{Ballesteros6} for the complex Bernoulli equations here studied, by deriving the corresponding change of variables.


\section{Exact solutions of complex Bernoulli equations}
\label{s3}
As mentioned in the introduction, our aim is to apply Proposition~\ref{PRO2} to the complex Bernoulli differential equations with $t$-dependent real coefficients, in order to obtain an explicit solution for any value of $n$.

 Let us consider the differential equation (\ref{be2})
\be
 \frac{\dd w}{\dd t}=a_1(t)w+a_2(t)w^n,\qquad n\notin\{0,1\},
 \label{c1}
\ee
where $w$ is a complex function, $n$ is a real number and $a_1(t),a_2(t)$ are arbitrary   real $t$-dependent parameters.   We introduce polar coordinates $(r,\theta)$ in the form   $w= r{\rm e}^{{\rm i}\theta}$, thus finding that the differential equation (\ref{c1}) leads to a first-order coupled system:
  \be
\begin{split}
\frac{\dd r}{\dd t} &=  a_1(t) r+a_2(t) r^n\cos[(n-1)\theta ],\\[4pt]
\frac{\dd \theta}{\dd t}&=  a_2(t)  r^{n-1} \sin[(n-1)\theta ].
\end{split}
\label{c2}
\ee
This system was already shown in ~\cite{Ballesteros6,BHLS}  to be an LH system, although its exact solution was not provided. According to the results in the preceding section, the system (\ref{c2}) is linearizable, i.e., it can be reduced via a change of variables to a linear system, although it is not immediately obvious how to find the appropriate new variables. It is precisely the LH formalism which provides a direct and systematic procedure to find such a reduction and hence, to solve the system explicitly. To this extent, we first reformulate the equations in terms of the $t$-dependent vector field
\be
     {\bf Y} = a_1(t) {\mathbf Y}_1+a_2(t) {\mathbf Y}_2,
\label{c3}
\ee
with
\be
{\>Y}_1=r\frac{\partial}{\partial r},\qquad  {\>Y}_2=r^n\cos[(n-1)\theta]\frac{\partial}{\partial r}+r^{n-1}\sin[(n-1)\theta]\frac{\partial}{\partial \theta},
\label{c4}
\ee
and commutator
\be
[{\>Y}_1, {\>Y}_2]=(n-1) {\>Y}_2,
\label{c5}
\ee
showing that ${\bf Y}$ determines a Lie system with Vessiot--Guldberg Lie algebra  isomorphic to $\mathfrak{b}_2$.  A compatible
symplectic form $\omega=f(r,\theta)\dd r\wedge \dd \theta$ is obtained imposing the relation  (\ref{b7b}), from which we deduce that
\be
 \omega = \frac{n-1}{r  \sin^2[(n-1)\theta ]}\,  \dd r\wedge \dd \theta .
 \label{c6}
 \ee
The associated Hamiltonian functions $  h_1$, $h_2$ (see relation (\ref{b6})) of the LH system are given by
\be
  h_1= -\frac{1}{\tan[(n-1)\theta]} \, , \qquad   h_2= -\frac{r^{n-1}}{\sin[(n-1)\theta]}\,  ,
 \label{c7}
\ee
while the corresponding Poisson bracket with respect to the symplectic form (\ref{c6})  reads
\be
\{   h_1,  h_2\}_\omega=-(n-1)  h_2 .
\label{c8}
\ee
As both LH systems (\ref{b9}) and (\ref{c2}) are based on the same Lie algebra $\mathfrak{b}_2$, there must exist a diffeomorphism between the expressions (\ref{b4}) with the canonical symplectic form (\ref{b3}) and variables $(x,y)$ and (\ref{c7}) with the non-canonical symplectic form (\ref{c6}) and the variables $(r,\theta)$. This change of variables is now easily found comparing the expressions for the Hamiltonian functions and symplectic form in both systems of coordinates, leading to the explicit expression
   \be
\begin{split}
x&=\frac{r^{n-1}}{\sin[(n-1)\theta]} \,  ,\qquad y=\frac{\cos[(n-1)\theta]}{(1-n)r^{n-1}}  \, ,\\[4pt]
r &=\left( \frac{x}{\sqrt{1+(n-1)^2 x^2 y^2}} \right)^{\frac{1}{n-1}} \,  ,\qquad  \theta =\frac 1{ (1-n)} \,
\arctan \left(\frac{1}{(n-1) xy}\right).
\end{split}
\label{c9}
\ee
The relationship between the corresponding vector fields, Hamiltonian functions and parameters is easily seen to be
  \be
\begin{split}
\>Y_1&=(n-1)\>X_A,\qquad \>Y_2=\>X_B,\qquad   h_1=(n-1) h_A,\qquad   h_2=h_B,\\[2pt]
a_1(t)&=b_A(t)/(n-1) ,\qquad a_2(t)=b_B(t).
\end{split}
\label{c10}
\ee
From these relations, an explicit solution of the Bernoulli equation (\ref{c2}) is obtained by merely introducing (\ref{b10}) into (\ref{c9}), where
$a_1(t)$,  $a_2(t)$ are arbitrary functions (and $n\notin\{0,1\}$):
   \be
\begin{split}
r(t) &=\left(  \frac{c_1 \,\eee^{\gam(t)}}{\sqrt{1+(n-1)^2 c_1^2\left( c_2 +  \int_a^t  \eee^{\gam(s )} a_2(s) \dd s\right)^2    }}  \right)^{\frac{1}{n-1}}   ,\qquad  \gam(t)=(n-1) \int_a^t a_1(\tau)\dd \tau, \\[4pt]
 \theta(t)& =-\frac 1{ (n-1)} \,
\arctan \left(\frac{1}{(n-1)  c_1\left( c_2 +  \int_a^t  \eee^{\gam(s )} a_2(s) \dd s\right)}\right) \,  . \\[4pt]
\end{split}
\label{c11}
\ee
It can be routinely checked that equations (\ref{c2}) are satisfied, observing that
\be
\sqrt{(n-1)^2 c_1^2\left( c_2 +  \int_a^t  \eee^{\gam(s )} a_2(s) \dd s\right)^2    }=-{(n-1) c_1\left( c_2 +  \int_a^t  \eee^{\gam(s )} a_2(s) \dd s\right)    },
\label{c12}
\ee
with the choice of the negative square root following from the relation
\be
\sqrt{(n-1)^2 x^2 y^2 }= - (n-1) xy = \frac 1{\tan[(n-1)\theta]},
\label{c13}
\ee
according to (\ref{c9}).

 Using the same argumentation as above, it is not difficult to verify that the differential equation
\be
 \frac{\dd w}{\dd t}=a_1(t)w+{\rm i} a_2(t)w^n,\qquad n\notin\{0,1\},
 \label{ca1}
\ee
with $w$ complex, can also be solved by quadratures. With the same change of coordinates $w= r{\rm e}^{{\rm i}\theta}$, the associated first-order system corresponding to the real and imaginary parts is given by
  \be
\begin{split}
\frac{\dd r}{\dd t} &=  a_1(t) r-a_2(t) r^n\sin[(n-1)\theta ],\\[4pt]
\frac{\dd \theta}{\dd t}&=  a_2(t)  r^{n-1} \cos[(n-1)\theta ].
\end{split}
\label{ca2}
\ee
In this case, the corresponding Hamiltonian vector fields  are
\be
{\>Z}_1=r\frac{\partial}{\partial r},\qquad  {\>Z}_2=-r^n\sin[(n-1)\theta]\frac{\partial}{\partial r}+r^{n-1}\cos[(n-1)\theta]\frac{\partial}{\partial \theta},
\label{ca4}
\ee
with commutator $[{\>Z}_1, {\>Z}_2]=(n-1) {\>Z}_2$. Indeed, the vector fields ${\bf Y}_2$ and ${\bf Z}_2$ are related through the orthogonal transformation
\begin{equation}
{\bf Z}_2 = \left(
\begin{array}[c]{cc}
0 & -r \\
r^{-1} & 0\end{array}\right) {\bf Y}_2,
\end{equation}
from which the assertion follows at once.

\section{Deformed Lie--Hamilton systems from the quantum  book algebra}
\label{s4}

Among the most interesting properties of LH systems, that distinguish them clearly from classical Lie systems, we enumerate the possibility of combining them with quantum groups, in order to obtain more general systems of differential equations. Even if such deformed systems are no longer described in terms of Lie algebras, we can still derive ``deformed superposition principles", using the formalism of Poisson--Hopf deformations of LH systems (see \cite{Ballesteros6,BCFHL} and references therein). In this section, the general solution by quadratures of the quantum deformations of $\mathfrak{b}_2$-systems is obtained, with the results being particularized to the deformed complex Bernoulli differential equations, for which a deformed superposition principle was previously studied in~\cite{Ballesteros6}.
 
We introduce the coboundary quantum deformation of the book Lie algebra $\mathfrak{b}_2$  (\ref{b1}) coming from the classical $r$-matrix
\be
r=z\, v_A\wedge v_B ,
\label{d1}
\ee
 which is a solution of the classical Yang--Baxter equation, and where $ z$ is the quantum deformation parameter such that $q={\rm e}^z$. A structure of Lie bialgebra is then determined through the cocommutator map $\delta: \mathfrak{b}_2 \to \mathfrak{b}_2 \wedge \mathfrak{b}_2$  obtained from the classical $r$-matrix (\ref{d1}) as (see e.g. \cite{CP})
 \begin{equation}
\delta(v_i)=[v_i\otimes 1+1\otimes v_i , r], \qquad i=A,B,
\label{d2}
\end{equation}
giving
\be
\delta(v_B)= 0,\qquad  \delta(v_A)=z\, v_B\wedge v_A .
\label{d3}
\ee
This element is just the  skew-symmetric part of the first-order term $\Delta_1$ in the quantum deformation parameter $z$  of the full coproduct $\Delta_z$, that is,
 \be
\begin{split}
\Delta_z(v_i)&=\Delta_0(v_i)+ \Delta_1(v_i)+o[z^2]  , \\[2pt]
 \delta(v_i)&=\Delta_1(v_i)-\sigma\circ \Delta_1(v_i) ,
\end{split}
\label{d4}
\ee
where $ \Delta_0(v_i)=v_i\otimes 1+ 1\otimes v_i $ is the  primitive (non-deformed) coproduct (\ref{b2}) and $\sigma$ is the flip operator: $\sigma(v_i\otimes v_j)=v_j\otimes v_i$.

 For the particular case of $\mathfrak{b}_2$, the quantum book  algebra $U_{   z}(\mathfrak{b}_2)\equiv \mathfrak{b}_{z,2 }$ is  defined by the following deformed coproduct, fulfilling the coassociativity  property (\ref{b2a}):
\be
\begin{split}
\Delta_z(v_A)&=v_A\otimes \eee^{-z v_B}+1\otimes v_A , \\[2pt]
\Delta_z(v_B)&=v_B\otimes 1 + 1 \otimes v_B .
\end{split}
\label{d5}
\ee
The compatible deformed commutation relation is thus given by
\begin{equation}\label{d5a}
[v_A,v_B]_z=- \frac{1-\eee^{-z v_B}}{z} .
\end{equation}
A deformed symplectic representation $D_z$ of $\mathfrak b_{z,2}$  in terms of  the canonical variables $(x,y)$ used in Section~\ref{s2} and possessing the same symplectic form (\ref{b3}) reads as
 \be
h_{z,A}:= D_z(v_A)=  \left( \frac{ \eee^{z x} -1}{z}\right) y,\qquad h_{z,B}:= D_z(v_B)=  - x ,
\label{d7}
\ee
closing on the following deformed   Poisson bracket with respect  to $\omega$:
\be
\{h_{z,A},h_{z,B}\}_\omega=-\frac{1-\eee^{-z h_{z,B}}}{z}\, .
\label{d8}
\ee
Using the relation (\ref{b6}), we easily obtain  the corresponding vector fields as
\be
  \>X_{z,A}=\left( \frac{ \eee^{z x}-1}{z} \right)\frac{\partial}{\partial x}- \eee^{z x} y \, \frac{\partial}{\partial y},
\qquad  \>X_{z,B}= \frac{\partial}{\partial y}.
\label{d9}
\ee
These vector fields span a smooth distribution, in the sense of  Stefan--Sussmann \cite{Ste,Sus}, through the commutator
 \be
 [{\bf X}_{z,A},{\bf X}_{z,B}]=  \eee^{z x} \, {\bf X}_{z,B},
 \label{d10}
 \ee
 and satisfy the invariance condition under the Lie derivative (\ref{b7b}) and the symplectic form $\omega$ (\ref{b3}) .

This construction leads to a deformed $t$-dependent Hamiltonian and  a $t$-dependent vector field depending on two real arbitrary parameters $b_A(t)$ and $b_B(t)$ as follows (compare with (\ref{b8})):
\begin{equation}
\begin{split}
h_z&=b_A(t)  h_{z,A} +b_B(t)  h_{z,B}   =  b_A(t)   \left( \frac{ \eee^{z x} -1}{z}\right) y - b_B(t) \,x   , \\[2pt]
  \>X_z&=b_A(t)   \>X_{z,A} +b_B(t)   \>X_{z,B}   =  b_A(t) \left( \left( \frac{ \eee^{z x}-1}{z} \right)\frac{\partial}{\partial x}- \eee^{z x} y \, \frac{\partial}{\partial y} \right)    + b_B(t) \, \frac{\partial}{\partial y}  ,
\end{split}
 \label{d11}
\end{equation}
giving rise to the following first-order non-autonomous system on $\mathbb R^2$:
\begin{equation}
\frac{\dd x}{\dd t}=  b_A(t) \left( \frac{ \eee^{z x}-1}{z} \right), \qquad
\frac{\dd y}{\dd t}= - b_A(t) \eee^{z x}   y+b_B(t).
 \label{d12}
\end{equation}
This system generalizes (\ref{b9}), in the sense that the expressions (\ref{d7})--(\ref{d12}) reduce to (\ref{b4})--(\ref{b9})  for $\lim z\to 0$.
Indeed, the deformed equation can be transformed into another independent on the deformation parameter $z$ by means of the change of variables \be x=-\frac{1}{z}\ln u, \qquad y=v,\ee leading to
\begin{equation}
\begin{split}
\frac{\dd u}{\dd t}&=  b_A(t) \left( u-1\right),\qquad
\frac{\dd v}{\dd t}= - \frac{b_A(t)}{u}  v+b_B(t).
\end{split}
 \label{d12a}
\end{equation}
These equations can be seen as the general ``canonical" system of differential equations for  the set of deformed LH systems based in the quantum  book algebra $\mathfrak b_{z,2}$, as given by (\ref{d5}). As the first of the equations is separable, the system can also be solved by quadratures. After a short computation, we arrive at the solution
\begin{equation}\label{can1}
u(t)=1+\tilde c_1 {\rm e}^{\gamma(t)},\qquad v(t)=\left(\tilde c_2+\int_a^t b_B(s)\,{\rm e}^{\sigma(s)} {\rm d}s\right){\rm e}^{-\sigma(t)},
\end{equation}
where
\begin{equation}
\sigma(t)=\int_a^t \frac{b_A(s)}{u(s)} {\rm d}s
\end{equation}
and $\gamma(t)$ is the same as given in (\ref{b10}). Equation (\ref{can1}) encompasses the solution of the quantum deformation for all LH systems based on the Lie algebra  $\mathfrak b_{2}$, provided that a proper diffeomorphism (change of variables) is found. We hence conclude that any quantum deformation of a  $\mathfrak b_{2}$-LH system can always be solved by quadratures, and is moreover explicitly integrable.

  In the original coordinates $x,y$, the exact solution is then given by
\be
\begin{split}
x(t)&=-\frac{\ln \left(1- z c_1 \,\eee^{\gam(t)} \right)}z,\qquad \tilde c_1=-z  c_1,\qquad \tilde c_2=  c_2,\\[4pt]
y(t)&=  \exp\left( -\int_a^t \frac {b_A(u)}{1- z c_1\,\eee^{\gam(u)}  } \, \dd u\right)  \left( c_2 +  \int_a^t     \exp\left(  \int_a^u \frac {b_A(v)}{1-  z c_1\,\eee^{\gam(v)}  } \, \dd v\right)  b_B(u) \dd u \right).
\end{split}
\label{d14}
\ee
Observe that the non-deformed limit $z\to 0$ is well defined, reducing to the exact solution (\ref{b10}) of the system (\ref{b9}). 


\subsection{On the approximations of order $k$ in the deformation parameter}
\label{s41}

In a certain sense, the presence of the quantum deformation parameter $z$  can be regarded as the introduction of a perturbation into  the initial LH system (\ref{b9}), in such a manner that  a nonlinear interaction or coupling between the  variables $(x,y)$ in the deformed LH system (\ref{d12}) arises  through the term $\eee^{z x} y$. This fact can be clearly appreciated by considering a power series expansion in $z$ and truncating it at a certain order. The resulting ODEs can be seen as an approximation to the deformed system.

  For the first-order, we obtain the system
\begin{equation}
\begin{split}
\frac{\dd x}{\dd t}&= b_A(t)  \bigl(  x  +\tfrac 12 z    x^2 \bigr) +o[z^2]  , \\[2pt]
\frac{\dd y}{\dd t}&=  - b_A(t) \bigl(   y +z\,     x y  \bigr)+b_B(t) +o[z^2]   ,
\end{split}
 \label{d13}
\end{equation}
which holds for a small value of   $z$. In this approximation, we find that   $z$ introduces a quadratic term $x^2 $ into the first equation,  becoming a (real) Bernoulli  equation, while  in the second, we get a nonlinear interaction term $xy$ (this does not alter the fact that the equation is linear in $y$). The change of variables \be x=\frac 1{z u}\ee reduces the Bernoulli equation to the linear equation
\begin{equation}
\frac{\dd u}{\dd t}= -b_A(t)  \bigl(  u  +\tfrac 12  \bigr),
\end{equation}
from which the solution by quadratures of the system follows at once. After a short algebraic manipulation, we arrive at the general solution for linear approximation in $z$, in the coordinates $x,y$:
\begin{equation}
\begin{split}
x(t)&=2 c_1 \left(2 {\rm e}^{-\gamma(t)}-z c_1\right)^{-1},\qquad
y(t)=  {\rm e}^{-\xi(t)} \left( c_2 +  \int_a^t    b_B(\tau) {\rm e}^{\xi(\tau)}   \dd \tau \right),\\
 \gam(t)&= \int_a^t b_A(s)\dd s , \qquad \displaystyle \xi(t)=\int_a^t b_A(s)\bigl(1+z x(s)\bigr)\, \dd s  .
\end{split}
\end{equation}

In the quadratic approximation, we get the system
\begin{equation}
\begin{split}
\frac{\dd x}{\dd t}&= b_A(t)  \bigl(  x  +\tfrac 12 z    x^2 +\tfrac{1}{6}z^2x^3\bigr) +o[z^3]  , \\[2pt]
\frac{\dd y}{\dd t}&=  - b_A(t) \bigl(   y +z\,     x y  +\tfrac{1}{2}z^2x^2 y\bigr)+b_B(t) +o[z^3].
\end{split}
 \label{d13b}
\end{equation}
The first equation is clearly of Abel type \cite{DAV}, and can be easily seen to admit the integrating factor
\begin{equation}
\mu= \left( z^2x^3+3zx^2+6x\right)^{-1},
\end{equation}
from which the first integral
\begin{equation}
\frac{1}{12}\ln \left(\frac{x^2}{z^2x^2+3zx+6}\right)-\frac{1}{2\sqrt{15}}\arctan\left(\frac{2z^2x+3z}{\sqrt{15}\,z}\right)-\frac{1}{6}\int_a^t b_A(s){\rm d}s = C
\end{equation}
is obtained. In general, even if the last integral can be explicitly solved, the solution of the ODE (hence, the system) cannot be written in explicit form. The same pattern holds for approximations at the order $k\geq 2$, leading to the separable equation
\begin{equation}
\frac{\dd x}{\dd t}= b_A(t)  \bigl(  x  +\tfrac 12 z    x^2 +\tfrac{1}{6}z^2x^3+\dots +\tfrac{1}{(k+1)!}z^{k}x^{k+1}\bigr) +o[z^{k+1}].
\end{equation}
 As the polynomial \be P_k(x)= 1  +\tfrac 12 z    x +\tfrac{1}{6}z^2x^2+\dots +\tfrac{1}{(k+1)!}z^{k}x^{k}\ee is irreducible over the rationals, having at most one real non-rational root, the application of the method of partial fractions or the Hermite method is not of great use in simplifying the integral $\displaystyle\int \frac{{\rm d}x}{x P_k(x)} $, without numerical approximations or integration by series. In any case, as the genus of the algebraic curve $y- xP_k(x)=0$ is zero for any $k\geq 1$, the integral is always rational, hence an elementary function, even if the solution cannot be presented in closed explicit form \cite{Whit,Hau}.

We conclude that, independently on the integral $ \int_a^t b_A(s){\rm d}s$, the system obtained from the $k^{th}$-order approximation cannot be solved explicitly, although it is integrable by quadratures.


\section{Deformed  complex Bernoulli differential equations}
\label{s5}

In order to construct the quantum deformation of the complex Bernoulli differential equations along the lines described in Section~\ref{s3}, we keep the same symplectic form (\ref{c6}), apply the change of variables (\ref{c9}) to the deformed Hamiltonian functions (\ref{d7}) and vector fields (\ref{d9}) in the variables $(x,y)$, and maintain the  relations (\ref{c10}). The deformed Hamiltonian functions in variables $(r,\theta)$  turn out to be
   \be
h_{z,1}= -\frac{\cos[(n-1)\theta]} {z\, r^{n-1}}\left(\exp\left(\frac{z\, r^{n-1}}{\sin[(n-1)\theta]}     \right)-1 \right),\qquad
 h_{z,2}= -\frac{r^{n-1}}{\sin[(n-1)\theta]}
\label{e1}
\ee
 such  that
\be
 \{   h_{z,1},  h_{z,2}\}_\omega=(n-1)\,\frac{ \eee^{-z   h_{z,2}}-1}{z}\, ,
\label{e2}
 \ee
 with respect to the non-canonical symplectic from (\ref{c6}). The corresponding deformed vector fields $\>Y_{z,i}$, related with  (\ref{e1}) by means of (\ref{b6}), are given by
   \be
\begin{split}
 {\>Y}_{z,1}& =\left( r\cos^2[(n-1)\theta]  \exp\left(  \frac{z\, r^{n-1}}{\sin[(n-1)\theta]}  \right) + \frac{  \sin^3[(n-1)\theta]}  {z\, r^{n-2}}  \left( \exp\left(  \frac{z\, r^{n-1}}{\sin[(n-1)\theta]}  \right) -1  \right)\right) \frac{\partial}{\partial r}    \\[2pt]
&     \quad\ + \sin^2[(n-1)\theta]\left( \frac{ \exp\left(\frac{z\, r^{n-1}}{\sin[(n-1)\theta]}     \right)}{\tan[(n-1)\theta]}- \frac{\cos[(n-1)\theta] }{z\, r^{n-1}}  \left( \exp\left(  \frac{z\, r^{n-1}}{\sin[(n-1)\theta]}  \right) -1  \right) \right) \frac{\partial}{\partial \theta} ,   \\[2pt]
 {\>Y}_{z,2}&=r^n\cos[(n-1)\theta]\frac{\partial}{\partial r}+r^{n-1}\sin[(n-1)\theta]\frac{\partial}{\partial \theta} ,
 \end{split}
\label{e3}
\ee
that satisfy the deformed commutator
\be
[ {\>Y}_{z,1}, {\>Y}_{z,2}]=(n-1)   \exp\left(  \frac{z\, r^{n-1}}{\sin[(n-1)\theta]}  \right)  {\>Y}_{z,2} \, ,
\label{e4}
\ee
again corresponding to a distribution in the Stefan--Sussmann sense. We next consider the deformed $t$-dependent Hamiltonian and vector field (\ref{d11}) expressed  in terms of the polar variables $(r,\theta)$ given in (\ref{c9}), thus finding  that the deformed Bernoulli system of differential equations is given by (see e.g.~\cite{Ballesteros6})
 \be
\begin{split}
 \frac{\dd r}{\dd t}&=  a_1(t) \!\left(\! r\cos^2[(n-1)\theta]  \exp\left(  \frac{z\, r^{n-1}}{\sin[(n-1)\theta]}  \right) + \frac{  \sin^3[(n-1)\theta]}  {z\, r^{n-2}}  \left( \exp\left(  \frac{z\, r^{n-1}}{\sin[(n-1)\theta]}  \right) -1  \right)\right)   \\[2pt]
& \quad    +a_2(t) r^n\cos[(n-1)\theta ], \\[2pt]
 \frac{\dd \theta}{\dd t}&=  a_1(t)  \sin^2[(n-1)\theta]\left( \frac{ \exp\left(\frac{z\, r^{n-1}}{\sin[(n-1)\theta]}     \right)}{\tan[(n-1)\theta]}- \frac{\cos[(n-1)\theta] }{z\, r^{n-1}}  \left( \exp\left(  \frac{z\, r^{n-1}}{\sin[(n-1)\theta]}  \right) -1  \right) \right)   \\[2pt]
 & \quad+a_2(t)  r^{n-1} \sin[(n-1)\theta ].
 \end{split}
\label{e5}
\ee
 Taking into account the general solution (\ref{d14}) for the deformed LH systems  with a quantum $\mathfrak{b}_2$-symmetry in (\ref{d12}), the change of variables (\ref{c9}), along with the relations (\ref{c10}), provide the exact solution for the deformed  complex Bernoulli equations (\ref{e5})  for arbitrary values of the parameters $a_1(t)$, $a_2(t)$ and coefficients $n\notin\{0,1\}$:
    \be
\begin{split}
r(t) &=\left(   -\frac{\ln \left(1- z c_1 \,\eee^{\gam(t)} \right)}z \right)^{\frac{1}{n-1}   }   \left\{  1+(n-1)^2\, \frac{\ln^2 \left(1- z c_1 \,\eee^{\gam(t)} \right)}{z^2 }    \right. \,   \\[4pt]
   \times& \left. \exp\left(\! -2\int_a^t \frac {(n-1)a_1(u)}{1- z c_1\,\eee^{\gam(u)}  } \, \dd u\right)  \left(\! c_2 +  \int_a^t     \exp\left(  \int_a^u \frac {(n-1)a_1(v)}{1- z c_1\,\eee^{\gam(v)}  } \, \dd v\right)  a_2(u) \dd u \right)^2  \right\}^{\frac{-1}{2(n-1)}   }   , \\[4pt]
 \theta(t)& =\frac {1}{ (n-1)}   \arctan \left(\frac{ \exp\left( \int_a^t \frac {(n-1)a_1(u)}{1- z c_1\,\eee^{\gam(u)}  } \, \dd u\right) }{(n-1) \, \frac{\ln \left(1- z c_1 \,\eee^{\gam(t)} \right)}z   \left( c_2 +  \int_a^t     \exp\left(  \int_a^u \frac {(n-1)a_1(v)}{1- z c_1\,\eee^{\gam(v)}  } \, \dd v\right)  a_2(u) \dd u \right)   }\right) \,  , \\[4pt]    \gam(t)&=(n-1) \int_a^t a_1(s)\dd s .
\end{split}
\label{e6}
\ee
It can be seen that the system (\ref{e5}) is properly satisfied. For the  corresponding calculations observe that (again, only one of the square roots is admissible)
   \be
\begin{split}  &\sqrt{(n-1)^2\, \frac{\ln^2 \left(1- z c_1 \,\eee^{\gam(t)} \right)}{z^2 }    \left( c_2 +  \int_a^t     \exp\left(  \int_a^u \frac {(n-1)a_1(v)}{1- z c_1\,\eee^{\gam(v)}  } \, \dd v\right)  a_2(u) \dd u \right)^2 }  \\[4pt]
 &\qquad =  (n-1)\, \frac{\ln \left(1- z c_1 \,\eee^{\gam(t)} \right)}{z }    \left(c_2 +  \int_a^t     \exp\left(  \int_a^u \frac {(n-1)a_1(v)}{1- z c_1\,\eee^{\gam(v)}  } \, \dd v\right)  a_2(u) \dd u \right)  ,
\end{split}
\label{e7}
\ee
which is a consequence of the relation (\ref{c13}).

  Finally, we note that, as expected, all   the above expressions reduce to those given in Section~\ref{s3}   under  the non-deformed limit $z\to 0$.


\section{Higher dimensional quantum deformations of $\mathfrak{b}_2$-LH systems}
\label{s6}

So far, we have revisited the exact solution for generic $\mathfrak{b}_2$-LH systems previously given in~\cite{CFH},  and applied this result to obtain the general solution for the complex Bernoulli equation \cite{Ballesteros6,BHLS}. In addition, we have also derived the general solution for the quantum deformation of $\mathfrak{b}_2$-LH systems and, as a byproduct, deduced the exact solution for the deformed complex Bernoulli equation considered in~\cite{Ballesteros6}. These results correspond to systems of differential equations in two variables,  with either the Cartesian $(x,y)$ or polar reference $(r,\theta)$. One of the remarkable features of the underlying (deformed) coalgebra structure of LH systems, determined by the coproduct map, is that it allows us to generalize them to higher dimensions \cite{Ballesteros6,BHLS}. In this section, we explicitly present the  two-dimensional counterpart of the previous results and analyze the strong differences between classical and quantum deformed models. The very same features will appear in any higher dimension.

Let us consider the primitive (non-deformed) coproduct $\Delta$ given by (\ref{b2}) on the tensor product $\mathbb{R}^2\otimes \mathbb{R}^2$. We construct  Hamiltonian functions $h^{(2)}_i$ on  $(\mathbb{R}^2)^2$  with Cartesian coordinates $(x_1,y_1, x_2,y_2)$  using the symplectic realization  $D$ in (\ref{b4}) (see e.g. \cite{Ballesteros6} for details) by means of the prescription
   \be
\begin{split}
h^{(2)}_A&:= (D\otimes D)(\Delta(v_A))=  x_1y_1 +   x_2y_2 ,\\[2pt]
h^{(2)}_B&:= (D\otimes D)(\Delta(v_B))=  - x_1- x_2 ,
\end{split}
\label{f1}
\ee
which obviously satisfies the Lie bracket (\ref{b5}) with respect to the symplectic form
\be
\omega = \dd x_1\wedge  \dd y_1+\dd x_2\wedge  \dd y_2.
\label{f3}
\ee
The associated Poisson bracket is then given by
 \be
\bigl\{   h^{(2)}_{A},  h_{B}^{(2)}\bigr\}_\omega=- h_B^{(2)} ,
\label{f2}
\ee
providing the two-dimensional Hamiltonian (compare with (\ref{b8}))
\be
h^{(2)}=b_A(t)  h^{(2)}_A +b_B(t)  h^{(2)}_B   =  b_A(t)\bigl(  x_1y_1 +x_2y_2\bigr)- b_B(t) \bigl( x_1+ x_2 \bigr).
\label{f4}
 \ee
From the point of view of differential equations, the system resulting from the latter Hamiltonian merely consists of two copies of the non-autonomous  ODEs  (\ref{b9}) on  $(\mathbb{R}^2)^2$, with the canonical variables $(x_1,y_1)$ and $(x_2,y_2)$.
 This property is common to all Lie systems, independently on the dimension of the symplectic representation, and implies that the general solution of such multidimensional systems is immediately obtained from the solution (\ref{b10}) given for the one-dimensional case.

The introduction of a quantum deformation, on the contrary, changes qualitatively the structure of the deformed system. The  quantum deformation of $\mathfrak{b}_2$, with a deformed coproduct $\Delta_z$ defined in (\ref{d5}),  as well as the realization $D_z$ given in (\ref{d7}), leads to a ``two-particle" symplectic realization with the following Hamiltonian functions (see \cite{Ballesteros6} for details):
    \be
\begin{split}
h^{(2)}_{z,A}&:= (D_z\otimes D_z)(\Delta_z(v_A))=  \left( \frac{ \eee^{z x_1} -1}{z}\right) \eee^{z x_2}y_1+\left( \frac{ \eee^{z x_2} -1}{z}\right) y_2 ,\\[2pt]
h^{(2)}_{z,B}&:= (D_z\otimes D_z)(\Delta_z(v_B))=  - x_1- x_2.
\end{split}
\label{f5}
\ee
With respect to the   symplectic form (\ref{f3}), we obtain the bracket
 \be
\bigl\{h^{(2)}_{z,A},h^{(2)}_{z,B}\bigr\}_\omega=-\frac{1-\eee^{-z h^{(2)}_{z,B}}}{z}\,
\label{f6}
\ee
and deformed Hamitonian (compare with Eq. (\ref{d11}))
    \be
\begin{split}
h_z^{(2)}&=b_A(t)  h^{(2)}_{z,A} +b_B(t)  h^{(2)}_{z,B}  \\[4pt]
&=
 b_A(t) \left(   \left( \frac{ \eee^{z x_1} -1}{z}\right) \eee^{z x_2}y_1+\left( \frac{ \eee^{z x_2} -1}{z}\right) y_2 \right)- b_B(t) \bigl( x_1+ x_2 \bigr).
\end{split}
\label{f7}
\ee
The corresponding equations of the motion are thus given by
 \be
\begin{split}
\frac{\dd x_1}{\dd t}&= b_A(t) \left(\frac{\eee^{zx_1}-1}{z}\right) \eee^{zx_2} ,  \qquad
\frac {\dd y_1}{\dd t}=- b_A(t)\,  \eee^{z(x_1+x_2)} y_1+b_B(t) ,\\[4pt]
\frac{\dd x_2}{\dd t}&= b_A(t)\left(\frac{\eee^{zx_2}-1}{z}\right) ,\qquad
 \frac {\dd y_2}{dt}=-b_A(t)\,\eee^{z x_2}\bigl((\eee^{z x_1 }-1)y_1+y_2\bigr) +b_B(t).
 \end{split}
\label{f8}
\ee
Clearly, the limit for $z\rightarrow 0$ leads to the undeformed (and uncoupled) 2-dimensional system. The main feature of the quantum deformed system is that it does not longer consist of two copies of the initial one-dimensional deformed system (\ref{d12}), but adds an additional nontrivial coupling of the variables. In this sense, the resulting differential equations (\ref{f8}) provide a different system that must be solved independently. This fact is a direct consequence of the non-trivial coproduct $\Delta_z$ in (\ref{d5}), which conveys interacting terms from its structure, being a general property for quantum deformations of LH systems (see  \cite{Ballesteros6} for further properties).

As expected, by means of the change of variables $x_k(t)=-\frac{1}{z}\ln u_k(t)$, the system (\ref{f8}) can be transformed into a linear system
\begin{equation}
\begin{split}
\frac{\dd u_1}{\dd t}&= \frac{b_A(t)}{u_2(t)} \left(u_1- 1\right),\qquad \frac {\dd y_1}{\dd t}=- b_A(t)\frac{y_1}{u_1\;u_2}+b_B(t) ,\\[4pt]
\frac{\dd u_2}{\dd t}&= b_A(t) \left(u_2- 1\right),\qquad
 \frac {\dd y_2}{dt}=\frac{b_A(t)}{u_2}\left(\frac{u_1-1}{u_1}\, y_1-y_2\right) +b_B(t),
 \end{split}
\label{f8a}
\end{equation}
that can again be easily solved by quadratures. We observe, in particular, that (\ref{f8a})   no longer depends on the deformation parameter $z$. In a certain sense, we can be consider this system as the prototype of a ``canonical form" for the quantum deformation in $z$.  Notice, in particular, that the first and third equations (second and fourth equations, respectively) are not copies, but differ in a genuine deformation term.

We conclude that, if one looks for higher dimensional Bernoulli equations in the non-deformed case, one will merely find a system consisting of copies of the undeformed ODE (see Section~\ref{s3}), while the deformed Bernoulli equations considered in  Section~\ref{s5} lead to new systems for each of the considered dimensions. This implies, in particular, that if an LH system can be solved explicitly, its multidimensional version will have the same property. Whether this holds for the quantum deformed systems, is still an open question.


\section{Conclusions and final remarks}
\label{s7} 

In this work we have analyzed the exact integrability of LH systems, enlarging previous work begun in \cite{CFH}. 
Although LH systems (and their quantum deformations) are primarily characterized by the possibility of obtaining a nonlinear (deformed)  superposition principle \cite{BHLS,Ballesteros6}, it has been shown that, for the case of the Lie algebra $\mathfrak{b}_2$, the geometric formalism can be used to obtain exact solutions of first-order systems, both at the classical and quantum deformed level, whenever an appropriate system of coordinates is found. Using a suitable local diffeomorphism, a realization by Hamiltonian vector fields is found that enables us to reduce the differential equations to a linear system, from which an explicit solution is obtained by quadratures. 
The generic form of first-order systems (locally) diffeomorphic to $\mathfrak{b}_2$-LH systems is found using the equivalent reformulation of the system in terms of Hamiltonian functions and the compatibility condition with the canonical symplectic form. In general, it is a routine task to check whether a given system of first-order ordinary differential equations defines a Lie system, and whether it admits a compatible symplectic structure that transforms it into an LH one. If such a system is known to admit a (local) system of coordinates that allows an integration by quadratures of the resulting differential equations, the approach proposed in this work may be  more adequate than a direct analysis of the system by other methods, at least computationally, as the solution in arbitrary coordinates can be deduced directly applying the inverse of a local diffeomorphism. The argument remains valid for   quantum deformations, even if formally they do not define an LH system any more, and where the deformation parameter can be seen as a perturbation term of the initial system that makes it possible to determine the approximation of $k^{th}$-order in the deformation parameter.  
The procedure has been applied to obtain exact solutions of the so-called complex Bernoulli equations, whenever they form an LH system. Besides the quantum deformation, the approximation of $k^{th}$-order in the deformation parameter has been studied, showing that it can be integrated by quadratures, although an explicit solution is no longer possible for orders $k\geq 3$. These results may be of interest for solving numerically these equations, as the integrals can always be approximated with standard methods. A question that arises naturally in this context is whether the resulting numerical approximation improves or is computationally more efficient than well-established techniques, such as the predictor-corrector or the Heun methods \cite{But}. A comparison of both approaches is worthy to be analyzed in detail.

Concerning the multidimensional deformation, it differs considerably from the multidimensional undeformed analogue, in the sense that genuine deformation terms are obtained and prevent that the resulting system merely consists of copies of the one-dimensional equation. In this situation, illustrated for the two-dimensional case, it has also been shown that the deformation is solvable by quadratures.

Besides the complex Bernoulli equation, using Proposition~\ref{PRO2}, the results on $\mathfrak{b}_2$-LH systems can also be used to determine exact solutions of other first-order systems strongly related to these equations. Consider for example the change of variables
\begin{equation}
x=u^pv^{-r},\qquad y=-v^q u^{-m} ,\qquad \Xi= \Lambda (u^{p-m-1}v^{q-r-1}),\qquad \Lambda=mr-pq\neq 0 .
\end{equation}
 With these new coordinates, the system (\ref{b9}) is transformed into (see Eq. (\ref{syta}))
\be
\begin{split}
\frac{\dd u}{\dd t}&= \frac{b_A(t)(r-q)}{\Lambda}\,u(t)+ \frac{r\,b_B(t)}{\Lambda\, v(t)^{q}}\,u(t)^{m+1},  \\[4pt]
\frac{\dd v}{\dd t}&= -\frac{b_A(t)(m-p)}{\Lambda}\,v(t)+ \frac{p\,b_B(t)u(t)^m}{\Lambda}\, v(t)^{1-q},\\[4pt]
\end{split}
\label{f18}
\ee
where the variables are non-trivially coupled. We observe that each of these equations, when considered separately from the other, corresponds formally to a real Bernoulli equation, so that system (\ref{f18}) can be interpreted as a coupled Bernoulli system.  Independently of the functions chosen for the coefficients $b_i(t)$, the system can be solved by quadratures using (\ref{b10}). Considering now the composition of the change of variables leading to the ``canonical" form (\ref{d12a}) of the quantum deformation of (\ref{b9}) with the change of variables defined above, we are led to the system
\be
\begin{split}
\frac{\dd u}{\dd t}&=b_A(t)\, \frac{(r+q)v(t)^r-q\,u(t)^p}{\Lambda\,u(t)^{p-1}}+ \frac{r\, b_B(t)}{\Lambda\,v(t)^{q}}\,u(t)^{m+1},  \\[4pt]
\frac{\dd v}{\dd t}&= -b_A(t)\,\frac{m\,u(t)^p-(m+p)v(t)^r}{\Lambda}\,v(t)+ \frac{p\,b_B(t)u(t)^m}{\Lambda}\,v(t)^{1-q},\\[4pt]
\end{split}
\label{f21a}
\ee
which   no longer corresponds to a coupling of Bernoulli equations for arbitrary values of $p,q,r,m$, although its general solution can always be obtained using Eq. (\ref{d14}). For some particular values, the resulting system also corresponds to the non-trivial coupling of Bernoulli equations. For example, for the values $p=-1$, $m=2$ and $r=-q$, the system (\ref{f21a}) reduces to
\be
\begin{split}
\frac{\dd u}{\dd t}&=  b_A(t)u(t)+ \frac{b_B(t)}{v(t)^q}\,u(t)^{3},  \\[4pt]
\frac{\dd v}{\dd t}&= \frac{2b_A(t)}{q \, u(t)}\, v(t)- \frac{b_A(t) }{q}\, v(t)^{1-q}+\frac{b_B(t)u(t)^2}{q}\, v(t)^{1-q},\\[4pt]
\end{split}
\label{f22a}
\ee
where both equations are still of Bernoulli type.

Potential applications or extensions of this work concern the analysis and exact solutions of second-order ordinary differential equations admitting $\mathfrak{b}_2$ as Vessiot--Guldberg algebra \cite{C132}, as well as their quantum deformation and multidimensional counterparts. Further, considering $\mathfrak{b}_2$ as a subalgebra of more general Lie algebras associated to LH systems, it is conceivable to extend the study of exact solutions to the quantum deformed systems of LH systems related to the oscillator algebra $\mathfrak{h}_4$ \cite{BHLS,LH2015,Ballesteros6}, either from the formal point of view, or taking into account specific applications \cite{CFH}.

In a general context, a problem not yet studied in detail is the systematic analysis of the solvability by quadratures, or the obtainment of criteria that ensure the existence of explicit solutions of LH systems and their (multidimensional) generalizations and deformations, as well as the generic description of systems that are (locally) equivalent to an LH system associated to a given Lie algebra. Another question to be settled is the precise identification of LH structures with systems of ordinary differential equations appearing in real-world applications, specifically with those systems possessing time-dependent coefficients, such as {\it e.g.} some specific types of Lotka--Volterra systems (see \cite{LH2015,BHLS} and references therein), where the procedure may provide an alternative approach to obtain either exact or approximate solutions, once the parameters have been appropriately adjusted. Work in these various directions is currently in progress.


\section*{Acknowledgements}

 \small
 
R.C.S.~and F.J.H.~have been partially supported by Agencia Estatal de Investigaci\'on (Spain)  under grant  PID2019-106802GB-I00/AEI/10.13039/501100011033.  F.J.H.~also acknowledges support  by the Regional Government of Castilla y Le\'on (Junta de Castilla y Le\'on, Spain) and by the Spanish Ministry of Science and Innovation MICIN and the European Union NextGenerationEU  (PRTR C17.I1).

\newpage

\small




\begin{thebibliography}{99}


\phantomsection
\addcontentsline{toc}{section}{References}


\bibitem{Page} J. M. Page. {\em Ordinary Differential Equations, with an Introduction to Lie's Theory of Groups of One Parameter}. (London: Macmillan \& Co.) 1897.

\bibitem{Pain} P. Painlev\'e. {\em Le\c{c}ons sur la th\'eorie analytique des \'equations diff\'erentielles profess\'ees \`a Stockholm}. (Paris: Hermann) 1897.

\bibitem{Gra} J. Gray. {\em Change and Variations: A History of Differential Equations to 1900}.  (New York: Springer) 2021.

 \bibitem{Arn} V. I. Arnol'd. {\em Geometrical Methods in the Theory of Ordinary Differential Equations}.  (New York: Springer) 1983.
    
\bibitem{Dua} L. G. S.~Duarte, S. E. S.~Duarte, L. A. C. P.~da Mota and J. E. F.~Skea.  Solving second-order ordinary differential equations by extending the Prelle-Singer method.  {\em J. Phys. A: Math. Gen.}  {\bf 34} (2001)  3015--3024.
\newblock \href {https://doi.org/10.1088/0305-4470/34/14/308}
  {\path{doi:10.1088/0305-4470/34/14/308}}

\bibitem{Chandra2}
V. K.~Chandrasekar, M.~Senthilvelan and M.~Lakshmanan.  On the complete integrability and linearization of nonlinear ordinary differential equations. III. Coupled first-order equations.  {\em Proc. R. Soc. A}  {\bf 465} (2009) 585--608.
 \newblock \href {https://doi.org/10.1098/rspa.2008.0239}
  {\path{doi:10.1098/rspa.2008.0239}}

\bibitem{Lie} S. Lie. {\em Vorlesungen \"uber Differentialgleichungen mit bekannten infinitesimalen Transformationen}. (Leipzig: B. G. Teubner) 1891.

\bibitem{Dik}  L. E.~Dickson. Differential equations from the group standpoint. {\em Annals Math.} {\bf 25} (1924)  287--378.
\newblock \href {https://doi.org/10.2307/1967773}
  {\path{doi:10.2307/1967773}}

 

\bibitem{Ince} E. L. Ince. {\em Ordinary Differential Equations}. (New York: Dover Publications Inc.) 1956.

\bibitem{Nuc} M. C. Nucci. The role of symmetries in solving differential equations. {\em  
Math. Comput. Modelling} {\bf 25} (1997) 181--193. 
 \newblock \href {https://doi.org/10.1016/S0895-7177(97)00068-X}
  {\path{doi:10.1016/S0895-7177(97)00068-X}}

\bibitem{Lak} M. Lakshmanan and S. Rajasekar. {\em Nonlinear Dynamics. Integrability, Chaos and Patterns}. (Berlin: Springer) 2003.
    
\bibitem{LSc} S.~Lie and G.~Scheffers. {\em Vorlesungen {\"u}ber continuierliche {G}ruppen mit geometrischen
  und anderen {A}nwendungen}. (Leipzig: Teubner) 1883.

\bibitem{Vess} E. Vessiot. Sur les syst\`emes d'\'equations diff\'erentielles du premier ordre qui ont des systèmes fondamentaux d'int\'egrales. {\em Annales Fac. Sci. Toulouse} {\bf 8} (1894)  H1--H33.

\bibitem{Wei} J. Wei and E. Norman. Lie algebraic solution of linear differential equations. {\em J. Math. Phys.} {\bf 4} (1963) 575--581.
\newblock \href {https://doi.org/10.1063/1.1703993}
  {\path{doi:10.1063/1.1703993}}    
    
    \bibitem{Reid} J. L. Reid and G. L. Strobel. The nonlinear superposition theorem of Lie and
Abel's differential equations. {\em Lett. Nuovo Cimento} {\bf 38} (1983) 448--452.
 \newblock \href {https://doi.org/10.1007/BF02789861}
  {\path{doi:10.1007/BF02789861}}

\bibitem{PW}
S. Shnider and P. Winternitz. Classification of systems of nonlinear ordinary differential equations with superposition principles. {\em J. Math. Phys.}  {\bf 25} (1984) 3155--3165.
\newblock \href {https://doi.org/10.1063/1.526085}
  {\path{doi:10.1063/1.526085}}

\bibitem{Bun} T. C. Bountis, V. Papageorgiou and P. Winternitz. On the integrability of systems of nonlinear ordinary differential equations with superposition
principles. {\em  J. Math. Phys.} {\bf 27}  (1986) 1215--1224.
\newblock \href {https://doi.org/10.1063/1.527128}
  {\path{doi:10.1063/1.527128}}
  
  
\bibitem{Car} J. F. Cari\~{n}ena, J. Grabowski and G. Marmo. Superposition rules, Lie theorem
and partial differential equations. {\em Rep. Math. Phys.} {\bf 60} (2007) 237--258.
 \newblock \href {https://doi.org/10.1016/S0034-4877(07)80137-6}
  {\path{doi:10.1016/S0034-4877(07)80137-6}}    
  
\bibitem{CGL}
 J. F.~Carine\~na,  J.~Grabowski  and J.~de Lucas. Lie families: theory and applications. {\em J. Phys. A:
Math. Theor.} {\bf  43} (2010)  305201.
  \newblock \href {https://doi.org/10.1088/1751-8113/43/30/305201}
  {\path{doi:10.1088/1751-8113/43/30/305201}}

\bibitem{Grun} A. M. Grundland and J. de Lucas. A Lie systems approach to the Riccati
hierarchy and partial differential equations. {\em J. Differ. Equ.} {\bf 263} (2017)
299--337.
 \newblock \href {https://doi.org/10.1016/j.jde.2017.02.038}
  {\path{doi:10.1016/j.jde.2017.02.038}}

\bibitem{Zoladek}
H.~\.Zo\l \c adek. The method of holomorphic foliations in planar periodic systems: the case of Riccati equations.  {\em J. Differ. Equ.}  {\bf 165} (2000)  143--173.
 \newblock \href {https://doi.org/10.1006/jdeq.1999.3721}
  {\path{doi:10.1006/jdeq.1999.3721}}

\bibitem{C131} R. Campoamor-Stursberg. Perturbations of Lagrangian systems based on the preservation of subalgebras of Noether
symmetries. {\it Acta Mech.} {\bf 227} (2016) 1941--1956.
\newblock \href {https://doi.org/10.1007/s00707-016-1621-6}
  {\path{doi:10.1007/s00707-016-1621-6 }}

\bibitem{Olv} P. J. Olver. {\em Applications of Lie Groups to Differential Equations}.  (New York: Springer) 1985.

\bibitem{Le3} P. G. L. Leach. Equivalence classes of second-order ordinary differential equations with only a three-dimensional Lie algebra of point symmetries and linearisation. {\it J. Math. Anal. Appl.} {\bf  284} (2003) 31--48.
\newblock \href {https://doi.org/10.1016/S0022-247X(03)00147-1}
  {\path{doi:10.1016/S0022-247X(03)00147-1}}
  
\bibitem{LuSa}
J.~de Lucas and  C.~Sard{\'o}n.
\emph{A Guide to Lie Systems with Compatible Geometric Structures}. (Singapore:
World Scientific)  2020.
\href {https://doi.org/10.1142/q0208}
  {\path{doi:10.1142/q02080}}

\bibitem{LH2015}
A.~Ballesteros, A.~Blasco, F.~J.~Herranz, J.~de Lucas and  C.~Sard{\'o}n. Lie--Hamilton systems on the plane: Properties, classification and applications.
 {\it J. Differ. Equ.} {\bf 258} (2015) 2873--2907.
 \href {https://doi.org/10.1016/j.jde.2014.12.031}
  {\path{doi:10.1016/j.jde.2014.12.031}}
 

\bibitem{BHLS} A.~Blasco, F.~J.~Herranz, J.~de Lucas and C.~Sard\'on.
{Lie--Hamilton systems on the plane: Applications and superposition rules.}
 {\it J. Phys. A: Math. Theor.}  {\bf 48} (2015)  345202.
  \href {https://doi.org/10.1088/1751-8113/48/34/345202}
  {\path{doi:10.1088/1751-8113/48/34/345202}}
  
\bibitem{BCFHL} A.~Ballesteros, R.~Campoamor-Stursberg, E.~Fern\'andez-Saiz, F.~J.~Herranz and J.~de Lucas.
Poisson--Hopf algebra deformations of Lie--Hamilton systems.
 {\em J. Phys. A: Math. Theor.} {\bf 51} (2018) 065202.
  \newblock \href {https://doi.org/10.1088/1751-8121/aaa090}
  {\path{doi:10.1088/1751-8121/aaa090}}
 
\bibitem{Ballesteros6}
A.~Ballesteros, R.~Campoamor-Stursberg, E.~Fern\'andez-Saiz, F.~J.~Herranz and J.~de Lucas.
{Poisson--Hopf deformations of Lie--Hamilton systems revisited: deformed superposition rules and applications to the oscillator algebra.}
{\it J. Phys. A: Math. Theor.} {\bf 54}  (2021) 205202.
  \href {https://doi.org/10.1088/1751-8121/abf1db}
  {\path{doi:10.1088/1751-8121/abf1db}}
 

\bibitem{CFH} R.~Campoamor-Stursberg, E.~Fern\'andez-Saiz and  F.~J.~Herranz.  Exact solutions and superposition rules for Hamiltonian systems generalizing time-dependent SIS epidemic models with stochastic fluctuations. {\it AIMS Math.} {\bf 8} (2023) 24025--24052.
\href{https://doi.org/10.3934/math.20231225}
  {\path{doi:10.3934/math.20231225}}

 \bibitem{BBHMR09}
A.~Ballesteros, A.~Blasco, F. J.~Herranz, F.~Musso and O.~Ragnisco. (Super)integrability from coalgebra symmetry: formalism and applications. {\it J. Phys. Conf. Ser.}   {\bf 175} (2009) 012004. 
\href{https://doi.org/10.1088/1742-6596/175/1/012004}
  {\path{doi:10.1088/1742-6596/175/1/012004}}

 
\bibitem{Abe}   M. S. ~Sweedler. {\em Hopf Algebras}. (New York: W. A. Benjamin Inc.)  1969.
 
\bibitem{CP}   V.~Chari and A.~Pressley. {\em A Guide to Quantum Groups}. (Cambridge: Cambridge University Press) 1994.

\bibitem{Kam} A. Gonz\'alez L\'opez, N. Kamran and P. J. Olver. Lie algebras of vector fields in the real plane. {\em Proc. London Math. Soc.} {\bf 64} (1992) 339--368.
\newblock \href {https://doi.org/10.1112/plms/s3-64.2.339}
  {\path{doi:10.1112/plms/s3-64.2.339}}

\bibitem{Buchdahl}
H. A.~Buchdahl.  A relativistic fluid sphere resembling the Emden polytrope of index 5.  {\em Astrophys. J.}  {\bf 140} (1964)  1512--1516.
 \newblock \href {https://doi.org/10.1086/148055}
  {\path{doi:10.1086/148055}}

\bibitem{Chandrasekar}
V. K.~Chandrasekar, M.~Senthilvelan and M.~Lakshmanan.  On the complete integrability and linearization of certain second-order nonlinear ordinary differential equations.  {\em Proc. R. Soc. A}  {\bf 461} (2005)  2451--2477.
 \newblock \href {https://doi.org/10.1098/rspa.2005.1465}
  {\path{doi:10.1098/rspa.2005.1465}}
  

\bibitem{Sus} H. J. Sussmann. Orbits of families of vector fields and integrability of distributions. {\it Trans. Amer. Math. Soc.} {\bf 180} (1973) 171--188.
\newblock \href {https://doi.org/10.2307/1996660}
  {\path{doi:10.2307/1996660}}    
  
  \bibitem{Ste} P. Stefan. Accessible sets, orbits and foliations with singularities. {\it Proc. London Math. Soc.} {\bf 29} (1974) 699--713.
\newblock \href {https://doi.org/10.1112/plms/s3-29.4.699}
  {\path{doi:10.1112/plms/s3-29.4.699}}


\bibitem{DAV} H. T.~Davis. {\em Introduction to Nonlinear Differential
and Integral Equations}. (New York: Dover) 1962.


\bibitem{Hau} W. Hauser and W. Burau. {\em Integrale algebraischer Funktionen und ebene algebraische Kurven}. (Berlin: VEB Deutscher Verlag der Wissenschaften) 1962.

\bibitem{Whit} E. T. Whittaker and G. N. Watson. {\em A Course in Modern Analysis}. (Cambridge: Cambridge Univ. Press) 1963.

\bibitem{But} J. C.~Butcher. {\em Numerical Methods for Ordinary Differential Equations}.  (New York: John Wiley \& Sons) 2003. 


\bibitem{C132} R. Campoamor-Stursberg. Low dimensional Vessiot-Guldberg-Lie algebras of second-order ordinary differential equations.  {\em Symmetry} {\bf 8} (2016)  15.
 \newblock \href {https://doi.org/10.3390/sym8030015}
  {\path{doi:10.3390/sym8030015}}
  
  
 


  
\end{thebibliography}
\end{document}